\documentclass{article}
\usepackage{spconf,amsmath,graphicx,bbm,array,dcolumn,epsfig,amsmath,amsfonts,yhmath,bm,amssymb,psfrag,color,soul,enumerate,stackengine,cite,amsthm}
\usepackage[dvipsnames]{xcolor}
\usepackage{tcolorbox}
\usepackage[ruled,vlined]{algorithm2e}
\usepackage{comment}
\usepackage{lipsum}

\usepackage[lining]{ebgaramond}

\usepackage{tikz}

\definecolor{lgreen} {RGB}{180,210,100}
\definecolor{ngreen} {RGB}{98,158,31}
\definecolor{dgreen} {RGB}{78,138,21}
\definecolor{MLOWLSgreen} {RGB}{0,140,130}
\definecolor{SDPpurple} {RGB}{191,0,191}
\definecolor{lred}   {RGB}{220,0,0}
\definecolor{nred}   {RGB}{224,0,0}
\definecolor{bred}   {RGB}{200,20,20}
\definecolor{nblue}  {RGB}{28,130,185}
\definecolor{jblue}  {RGB}{20,50,100}

\newcommand {\myvec}[1] {{\mbox{\boldmath $#1$}}}
\newcommand {\mymat}[1]  {{\mbox{\boldmath $#1$}}}

\DeclareMathAlphabet      {\mathbfit}{OML}{cmm}{b}{it}

\AtBeginEnvironment{bmatrix}{\setlength{\arraycolsep}{4pt}}

\newcommand {\mS} {\mymat{S}}

\newcommand {\mSigma} {\mymat{\Sigma}}

\newcommand {\Omeg} {\mymat{\Omega}}
\newcommand {\U} {\mymat{U}}

\newcommand {\bS} {\mybar{\mS}}
\newcommand {\bX} {\mybar{\X}}

\newcommand {\D} {\mymat{D}}

\newcommand {\upsi} {\myvec{\psi}}

\renewcommand {\H} {\mymat{H}}

\newcommand {\PI} {\mymat{\Pi}}

\newcommand {\I} {\mymat{I}}
\newcommand {\X} {\mymat{X}}

\newcommand {\ub} {\myvec{b}}

\newcommand {\up} {\myvec{p}}

\newcommand {\uv} {\myvec{v}}

\newcommand {\uo} {\myvec{0}}
\newcommand {\us} {\myvec{s}}

\newcommand {\ulambda} {\myvec{\lambda}}

\newcommand {\ux} {\myvec{x}}

\newcommand {\uh} {\myvec{h}}

\newcommand {\uz} {\myvec{z}}

\newcommand {\uvarep} {\myvec{\varepsilon}}

\newcommand {\Rset} {\mathbb{R}}
\newcommand {\Cset} {\mathbb{C}}

\newcommand {\Eset} {\mathbb{E}}

\newcommand {\Diag} {\text{\normalfont Diag}}
\newcommand {\Blkdiag} {\text{\normalfont Blkdiag}}
\newcommand {\tps} {\rm{T}}
\newcommand {\her} {\rm{H}}

\newcommand {\hup} {\widehat{\up}}

\newcommand {\env} {\text{\tiny env}}
\newcommand {\col} {\text{\tiny col}}
\newcommand {\row} {\text{\tiny row}}
\newcommand {\snr} {\mathtt{SNR}}

\newcommand {\sbar} {\mybar{s}}
\newcommand {\xbar} {\mybar{x}}
\newcommand {\vbar} {\mybar{v}}

\newcommand {\buv} {\mybar{\uv}}
\newcommand {\bus} {\mybar{\us}}

\newcommand {\bux} {\mybar{\ux}}
\newcommand {\ud} {\myvec{d}}

\makeatletter
\newsavebox\myboxA
\newsavebox\myboxB
\newlength\mylenA

\newcommand*\mybar[2][0.75]{%
	\sbox{\myboxA}{$\m@th#2$}%
	\setbox\myboxB\null
	\ht\myboxB=\ht\myboxA%
	\dp\myboxB=\dp\myboxA%
	\wd\myboxB=#1\wd\myboxA
	\sbox\myboxB{$\m@th\overline{\copy\myboxB}$}
	\setlength\mylenA{\the\wd\myboxA}
	\addtolength\mylenA{-\the\wd\myboxB}%
	\ifdim\wd\myboxB<\wd\myboxA%
	\rlap{\hskip 0.5\mylenA\usebox\myboxB}{\usebox\myboxA}%
	\else
	\hskip -0.5\mylenA\rlap{\usebox\myboxA}{\hskip 0.5\mylenA\usebox\myboxB}%
	\fi}
\makeatother

\newtheorem{prop}{Proposition}

\newtheorem{lem}{Lemma}
\newtheorem{remark}{Remark}

\newtheorem{theorem}{Theorem}

\title{Towards Robust Data-Driven Underwater Acoustic Localization: \\ A Deep CNN Solution with Performance Guarantees for Model Mismatch}

\name{Amir Weiss$^{\star}$, Andrew C.\ Singer$^{\dagger}$, and Gregory W. Wornell$^{\star}$}

\address{
\begin{tabular}{ccc}
$^{\star}$Research Laboratory of Electronics &  $^\dagger$Dept. of Electrical and Computer Engineering\\
Massachusetts Institute of Technology & University of Illinois Urbana-Champaign\\
\{amirwei,gww\}@mit.edu & acsinger@illinois.edu
\end{tabular}
\thanks{
This work was supported, in part, by ONR under Grants No.\ N00014-19-1-2665 and No.\ N00014-19-1-2662, and NSF under Grant No.\ CCF-1816209.
}}

\begin{document}
\ninept
\maketitle

\begin{abstract}
\vspace{-0.1cm}
\small{Key challenges in developing underwater acoustic localization methods are related to the combined effects of high reverberation in intricate environments. To address such challenges, recent studies have shown that with a properly designed architecture, neural networks can lead to unprecedented localization capabilities and enhanced accuracy. However, the robustness of such methods to \emph{environmental mismatch} is typically hard to characterize, and is usually assessed only empirically. In this work, we consider the recently proposed data-driven method [19] based on a deep convolutional neural network, and demonstrate that it can learn to localize in complex and \emph{mismatched} environments. To explain this robustness, we provide an upper bound on the localization mean squared error (MSE) in the ``true" environment, in terms of the MSE in a ``presumed" environment and an additional penalty term related to the environmental discrepancy. Our theoretical results are corroborated via simulation results in a rich, highly reverberant, and mismatch channel.}
\end{abstract}
\vspace{-0.3cm}
\section{Introduction}\label{sec:intro}
\vspace{-0.2cm}

Localization in the underwater acoustic (UWA) domain, which is relevant for a variety of important applications \cite{waterston2019ocean,sun2020underwater,tuna2017survey,erol2010localization}, is challenging due to the environment's physical properties. In particular, the channel between an acoustic source and a receiver may vary dramatically with their positions. While the related literature is abundant with handcrafted model-based solutions \cite{tan2011survey}, many of these methods rely on isovelocity line-of-sight-based geometry, and are thus sensitive to multipath. Other methods that can theoretically operate in nontrivial environments, such as matched field processing \cite{baggeroer1993overview}, are highly sensitive to model mismatch \cite{collins1991focalization}, and are therefore less relevant for some practical purposes. 

Due to the success of deep neural networks (DNN) in other domains, data-driven (DD) methods have been recently considered more extensively as potentially viable solutions to the UWA localization problem \cite{bianco2019machine,niu2019deep,testolin2019underwater,gong2020machine,chen2021model,lefort2017direct,wang2018underwater,qin2020underwater,houegnigan2017machine}. Specifically, a DD variant of the statistically superior direct localization (DLOC) approach \cite{weiss2004direct,wang2020direct} has been proposed in \cite{weiss2022direct}. Yet, while these DD methods are in principle exceptionally powerful, they can be difficult to analyze, and their robustness is generally not guaranteed.

This work is focused on the robustness of a DD-DLOC solution to UWA localization, and its main contribution is an upper bound on the localization mean squared error (MSE) in a mismatched environment. To better contextualize and elaborate on our contributions (listed at the end of Section \ref{sec:problemformulation}), let us first describe in detail the problem at hand. 


\vspace{-0.2cm}
\section{Problem Formulation}\label{sec:problemformulation}
\vspace{-0.2cm}
Consider $L$ spatially-diverse, time-synchronized receivers at known locations, each consisting of a single omni-directional sensor. Furthermore, consider the presence of a source, emitting an unknown waveform from an unknown position, denoted by the vector of coordinates $\up\in\mathcal{V}\subset\Rset^{3\times 1}$, where $\mathcal{V}$ is a volume of interest. We assume that the source is static during the observation time, and is located sufficiently far from all $L$ receivers to permit a planar wavefront (far-field) approximation.\footnote{This becomes reasonably accurate, e.g., in shallow waters at high frequencies.}

In the underwater environment, ``reflectors" at unknown positions are potentially present, giving rise to a (possibly rich) multipath channel impulse response (CIR), which can be approximately described using ray propagation (e.g., \cite{etter2018underwater}). For example, the water surface acts as one of the reflectors. In such a model, the acoustic wave that is emitted from the source and measured at the the receiver is represented as a sum of (possibly infinitely many) rays, each propagating according to physical properties of the environment (e.g., temperature, bathymetry, etc.). For simplicity, we assume the existence of a set of the environmental parameters, denoted by $\mathcal{P}_{\env}$, with which the propagation model is fully characterized.



Formally, the received baseband signal of the $\ell$-th receiver is given by
\begin{equation}\label{modelequation}
\begin{gathered}
x_{\ell}[n]=\sum_{r=1}^{R}b_{r\ell}s_{r\ell}[n]+v_{\ell}[n]\triangleq\us_{\ell}^{\tps}[n]\ub_{\ell}+v_{\ell}[n]\in\Cset,\\
n=1,\ldots,N, \;\forall \ell\in\{1,\ldots,L\},
\end{gathered}
\end{equation}
where we have defined $\us_{\ell}[n]=[s_{1\ell}[n] \cdots s_{R\ell}[n]]^{\tps}\in\Cset^{R\times1}$ and $\ub_{\ell}=[b_{1\ell} \cdots b_{R\ell}]^{\tps}\in\Cset^{R\times1}$, using the notation:
{\begin{enumerate}
    \setlength{\itemsep}{0pt}
	\item $b_{r\ell}\in\Cset$ as the unknown attenuation coefficient\footnote{In fact, these coefficients are also a function of the source's position $\up$ and the environment $\mathcal{P}_{\env}$, but we mostly use $b_{r\ell}$ rather than $b_{r\ell}(\up,\mathcal{P}_{\env})$ for brevity.} from the source to the $\ell$-th sensor associated with the $r$-th signal component (line-of-sight (LOS) or other non-LOS (NLOS) reflections);
	\item $s_{r\ell}[n]\triangleq \left.s\left(t-\tau_{r\ell}(\up,\mathcal{P}_{\env})\right)\right\vert_{t=nT_s}\in\Cset$ as the sampled $r$-th component of the unknown signal's waveform at the $\ell$-th sensor, where $s\left(t-\tau_{r\ell}(\up,\mathcal{P}_{\env})\right)$ is the analog, continuous-time waveform delayed by $\tau_{r\ell}(\up,\mathcal{P}_{\env})$, and $T_s$ is the sampling period; and
	\item $v_{\ell}[n]\in\Cset$ as the additive, ambient and internal receiver, noise at the $\ell$-th receiver, modeled as a zero-mean random process (not necessarily Gaussian) with an unknown finite variance $\sigma_{v_{\ell}}^2$.
\end{enumerate}}

Applying the normalized DFT\footnote{$\mybar{\uz}$ denotes the normalized discrete Fourier transform (DFT) of $\uz$.} to \eqref{modelequation} yields the equivalent frequency-domain representation for all $\ell\in\{1,\ldots,L\}$,
\begin{equation}\label{modelequationfreq}
\begin{aligned}
\hspace{-0.2cm}\xbar_{\ell}[k]&=\sum_{r=1}^{R}b_{r\ell}\sbar[k]e^{-\jmath\omega_k\tau_{r\ell}({\text{\boldmath $p$}},\mathcal{P}_{\env})}+\vbar_{\ell}[k]\\
&=\sbar[k]\hspace{-0.05cm}\cdot\hspace{-0.05cm}\underbrace{\ud_{\ell}^{\her}[k]\ub_{\ell}}_{\substack{\text{a function of} \\ \text{{\boldmath $p$} and $\mathcal{P}_{\env}$} }}+\vbar_{\ell}[k]=\sbar[k]\hspace{-0.05cm}\cdot\hspace{-0.05cm}\underbrace{\bar{h}_{\ell}[k]}_{\substack{\text{frequency} \\ \text{response}}}+\vbar_{\ell}[k]\in\Cset,
\end{aligned}
\end{equation}
where we have defined $\bar{h}_{\ell}[k]\triangleq\ud_{\ell}^{\her}[k]\ub_{\ell}\in\Cset$, and
\begin{equation*}
\ud_{\ell}[k]\triangleq[e^{-\jmath\omega_k\tau_{1\ell}({\text{\boldmath $p$}},\mathcal{P}_{\env})} \cdots e^{-\jmath\omega_k\tau_{R\ell}({\text{\boldmath $p$}},\mathcal{P}_{\env})}]^{\her}\in\Cset^{R\times 1},\label{dft_freq_vec}
\end{equation*}
with $\{\omega_k\triangleq\frac{2\pi(k-1)}{NT_s}\}_{k=1}^N$. As shorthand notation, we further define
\begin{equation*}
\begin{aligned}
\bux_{\ell}\hspace{-0.025cm}&\triangleq\hspace{-0.025cm}\left[\mybar{x}_{\ell}[1] \cdots \mybar{x}_{\ell}[N]\right]^{\tps}, \bus\triangleq \left[\sbar[1] \cdots \sbar[N]\right]^{\tps},\\
\buv_{\ell}\hspace{-0.025cm}&\triangleq\hspace{-0.025cm} \left[\mybar{v}_{\ell}[1] \cdots \mybar{v}_{\ell}[N]\right]^{\tps}\hspace{-0.05cm}, \D_{\ell}\triangleq\left[\ud_{\ell}[1] \cdots \ud_{\ell}[N]\right]^{\tps}\hspace{-0.05cm}\in\Cset^{N\times R},\\
\bX_{\ell}\hspace{-0.025cm}&\triangleq\hspace{-0.025cm}\Diag(\bux_{\ell}), \bS\triangleq\Diag(\bus), \H_{\ell}(\up,\mathcal{P}_{\env})\triangleq\Diag\left(\D_{\ell}\ub_{\ell}\right),
\end{aligned}
\end{equation*}
where $\Diag(\cdot)$ forms a diagonal matrix from its vector argument. Note that $\H_{\ell}(\up,\mathcal{P}_{\env})$ is generally a nonlinear function of the unknown source position $\up$ and the environmental parameters $\mathcal{P}_{\env}$ (e.g., as seen from \eqref{dft_freq_vec}). With this notation, we may now write \eqref{modelequationfreq} compactly as
\begin{equation}\label{signalmodelfreqcompact}
\bux_{\ell}=\H_{\ell}(\up,\mathcal{P}_{\env})\bus+\buv_{\ell}\in\Cset^{N\times 1}, \;\forall\ell\in\{1,\ldots,L\}.
\end{equation}

Given the possibly (highly) intricate relation between $\up$ and $\bux_{\ell}$, the complexity of the localization problem, which consists of estimating $\up$ from the observations $\left\{\bux_{\ell}\right\}_{\ell=1}^{L}$, is now readily apparent from \eqref{signalmodelfreqcompact}.

\begin{figure}[t]
	\includegraphics[width=0.48\textwidth]{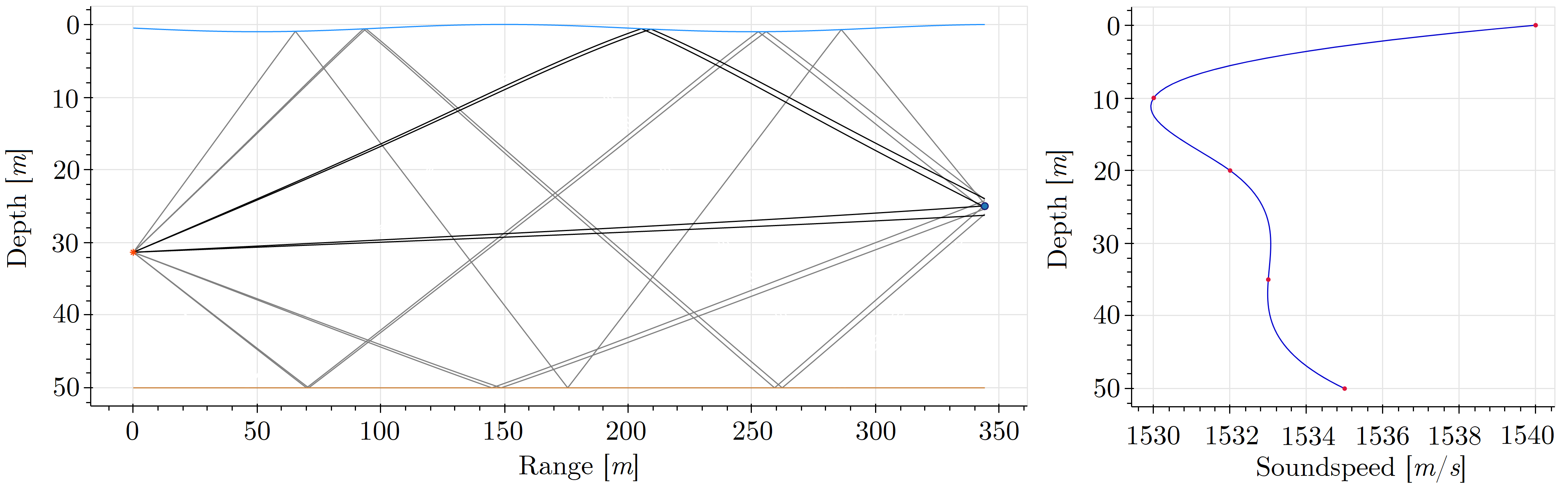}
	\centering\vspace{-0.3cm}
	\caption{Illustration of the simulated environment using Bellhop \cite{porter2011bellhop}. Left: Ray propagation paths from the source (red) to a receiver (blue). Right: The assumed depth-varying speed of sound profile.}
	\label{fig:example}\vspace{-0.5cm}
\end{figure}
To further illustrate the potentially formidable computational task of localizing a source via an exact model-based solution in an environment with nontrivial propagation properties, consider the UWA environment illustrated in Fig.~\ref{fig:example}. Given a source-receiver pair, the relevant rays propagating from the source to the receiver can be described in a plane perpendicular to an ideal flat bottom, in which the source and the receiver lie. As an example, the complexity of the method \cite{weiss2022semi} in its extended form \cite{weiss2022direct} is $\Omega(NL^2R^2)$ even before considering that all the $LR$ time-delays of all the rays to all the receivers must be computed as well.\footnote{For environments such as the on in Fig.~\ref{fig:example}, this may already be infeasible.}

It is therefore natural to adopt a learning-based solution approach, that will provide a method to localize the source while: (i) judiciously using those environmental characteristics that can be learned during the training phase; (ii) being robust to reasonably moderate deviations from the learned model; and (iii) efficiently approximating a reliable estimator $\hup=f(\bux_{1},\ldots,\bux_{L})$ that is otherwise computationally infeasible.

The remainder of this paper is devoted to the study of the desired merits of a learning-based solution as mentioned above. Specifically, 
\begin{itemize}
    \item We demonstrate that a DD approach can yield a computationally efficient DLOC solution that is ``stable"---in the sense of its MSE performance---with respect to model mismatch deviations. 
    \item We provide analytical guarantees on the MSE localization performance of an estimator designed for one environment, when applied to data from a another environment (i.e., model mismatch).
\end{itemize}

\vspace{-0.3cm}
\section{Learning to Localize}\label{sec:learning2localize}
\vspace{-0.2cm}
Let us assume that we do not have precise knowledge of the environment $\mathcal{P}_{\env}$ or the propagation model, based on which $\H_{\ell}(\up,\mathcal{P}_{\env})$ is computed for any $\up$. Rather, assume that a dataset of received signals with labeled source positions $\mathcal{D}^{(J)}_{\text{train}}\triangleq\{(\bux^{(i)}_{1},\ldots,\bux^{(i)}_{L},\up^{(i)})\}_{i=1}^{J}$ of size $J$ is available. Such a dataset can be obtained in several ways (see \cite{weiss2022direct} for details).
\vspace{-0.2cm}
\subsection{A Data-Driven Direct Localization Method}
\vspace{-0.1cm}
In a recent work \cite{weiss2022direct}, an architecture based on a deep convolutional neural network (CNN) has been proposed as a DD-DLOC solution for the UWA localization problem. This architecture is trained in two phases. In the first phase, three ``sub-models" are trained to estimate range, azimuth and inclination. In the second phase, these three models are fused into a ``global model", that provides an (enhanced) estimate of the source position. The successful operation of this solution has been demonstrated for a $3$-ray isovelocity propagation model, which theoretically allows for accurate DLOC even in the absence of LOS (see \cite{weiss2022semi} for simulation and experimental results). While this is already a powerful capability, in this work we demonstrate that this solution generalizes to even richer environments, which are ``closer"---in some well-defined sense---to the real physical medium in which UWA localization systems operate.

To this end, in this work we consider a simulated UWA environment as our virtual testbed, and adopt the architecture proposed in \cite{weiss2022direct}, to which we introduce a two main modifications. First, we now directly concatenate the outputs of the three sub-models, such that the modified architecture already provides a working solution at the end of the first phase of training (mentioned above). Second, during training the CNN-based solution is further matched to the attenuation law of the channel, by normalizing the signal power with the average attenuation magnitude in the volume of interest. For more implementation details of this architecture, see the supplementary material for this paper, \cite{weiss2022icasspsuppmat}.

To go beyond the $3$-ray model \cite{weiss2022semi}, we generate the CIRs using the Bellhop simulator \cite{porter2011bellhop}. This way, the propagation model can be made exceptionally rich, taking into account nontrivial affects such as ``bending" rays due to depth-varying speed of sound, different bathymetry and surface geometries, etc. Thus, the simulated received signals are more similar to true recorded signals in a real environment. For more details, see \cite{weiss2022icasspsuppmat}.

As we demonstrate in Section \ref{sec:results}, our architecture is able to capture the environment's characteristics relevant for inference regarding the source position. While this is a significant part of the contribution of this work, we nonetheless choose to elaborate on yet another important aspect---robustness to model mismatch. Although a carefully designed simulator (such as Bellhop) can provide good approximations of the real CIRs, it is still only an \emph{approximation} that naturally suffers (to some extent) from modeling error relative to the respective true CIRs as well as unmodelled effects. For example, a simulator may provide a CIR with a finite number of rays, when in practice there may be an infinite number (with finite energy). Therefore, and especially given that analytical characterization of the output of DNNs is generally a hard task, we next turn to the main discussion of this paper---an upper bound on the MSE of the proposed CNN-based solution under model mismatch.

\vspace{-0.2cm}
\section{Performance Guarantees Under a Mismatched Propagation Model}
\vspace{-0.2cm}
For simplicity of the exposition, we assume hereafter that $\sigma_{v_{\ell}}^2=\sigma_{v}^2$ for all $\ell\in\{1,\ldots,L\}$. However, this assumption can be relaxed, and with the necessary slight modifications, the results below are still valid.

The signal model \eqref{modelequation} more faithfully characterizes the true underlying channel, e.g., when $R\to\infty$, in which case $\sum_{r=1}^{\infty}|b_{r\ell}|^2<\infty$ for all $\ell\in\{1,\ldots,L\}$, when the power of the received signal is finite. However, practical localization solutions would take into account only a finite number of NLOS components (if, at all), and would therefore suffer from model mismatch. Nevertheless, if $R$ is chosen to be sufficiently large, such that most of the energy of the received signal is captured by the model, it is reasonable to expect that the devised localization method would still perform properly, with some slight performance degradation relative to the expected performance in the absence of model mismatch. Below we provide a rigorous proof that this intuition is indeed correct. Moreover, we quantify ``which portions" of the channel must be modeled in order to still maintain MSE performance guarantees.

\vspace{-0.3cm}
\subsection{Upper Bounds on the MSE for Mismatched Environment}\label{subsec:upperbounds}
\vspace{-0.1cm}
To facilitate the derivations below, we introduce a compact representation of \eqref{signalmodelfreqcompact}. Let $\bux\triangleq[\bux_1^{\tps} \ldots \bux_L^{\tps}]^{\tps}\in\Cset^{NL\times 1}$, with which \eqref{signalmodelfreqcompact} reads
\begin{equation}\label{compactmodelforCRLB}
\bux=\H(\up,\mathcal{P}_{\env})\bus+\buv\in\Cset^{NL\times 1},
\end{equation}
where $\H(\up,\mathcal{P}_{\env})\triangleq\left[\H^{\tps}_{1}(\up,\mathcal{P}_{\env}) \ldots \H^{\tps}_{L}(\up,\mathcal{P}_{\env})\right]^{\tps}\in\Cset^{NL\times N}$ and $\buv\triangleq[\buv_1^{\tps} \ldots \buv_L^{\tps}]^{\tps}$. Further, let $Q\triangleq\mathcal{CN}(\uo,\mSigma_{Q})$, $P\triangleq\mathcal{CN}(\uo,\mSigma_{P})$ be the distributions of the presumed and the true\footnote{Possibly with $R\to\infty$, or a different type of model mismatch.} models for \eqref{compactmodelforCRLB} with environmental parameters $\mathcal{Q}_{\env}, \mathcal{P}_{\env}$, respectively, where
\begin{equation}\label{defofcovariancematrices}
\begin{aligned}
    \mSigma_{Q}&\triangleq\Eset_Q\left[\bux\bux^{\her}\right]=\sigma_s^2\H(\up,\mathcal{Q}_{\env})\H^{\her}(\up,\mathcal{Q}_{\env})+\sigma_v^2\I_{NL}, \\ \mSigma_{P}&\triangleq\Eset_P\left[\bux\bux^{\her}\right]=\sigma_s^2\H(\up,\mathcal{P}_{\env})\H^{\her}(\up,\mathcal{P}_{\env})+\sigma_v^2\I_{NL}.
\end{aligned}
\end{equation}

Let $\hup:\bux\to\Rset^{K_p\times1}$ be a location estimator of $\up$, and define $\uvarep_{\text{\boldmath $p$}}\triangleq\hup-\up$, its respective estimation error vector. We have the following result:
\vspace{-0.5cm}
\begin{prop}\label{prop0}
Denote $Q_{\varepsilon}$ and $P_{\varepsilon}$ as the distributions of $\uvarep_{\text{\boldmath $p$}}$ when $\bux$ is distributed according to $Q$ and $P$, respectively. Then the MSE of $\hup$ evaluated for the true model $P$ is upper bounded by
\begin{align}
    \Eset_{P}\left[\|\uvarep_{\text{\boldmath $p$}}\|^2\right]&\leq \Eset_{Q}\left[\|\uvarep_{\text{\boldmath $p$}}\|^2\right]+\delta(P_{\varepsilon},Q_{\varepsilon}),\label{strongupperbound}
\end{align}
where $\delta(P_{\varepsilon},Q_{\varepsilon})\triangleq\sqrt{{\rm Var}_{Q}\left(\|\uvarep_{\text{\boldmath $p$}}\|^2\right)\chi^2\left(P_{\varepsilon}||Q_{\varepsilon}\right)}$, and $\chi^2\left(P_{\varepsilon}||Q_{\varepsilon}\right)$ denotes the chi-square divergence (CSD) of $P_{\varepsilon}$ from $Q_{\varepsilon}$.
\end{prop}
\begin{proof}
The bound \eqref{strongupperbound} is obtained by directly applying \cite[Th.~1]{weiss2022bilateralbound}.
\end{proof}
\begin{remark}
The bound \eqref{strongupperbound} is agnostic to the type/form of model mismatch.
\end{remark}
The meaning of Prop.~\ref{prop0} is the following. If for the presumed environment $\mathcal{Q}_{\env}$, which is used for training of the neural network, the localization error is $Q_{\varepsilon}$-distributed, and this distribution is ``not too far"---in the CSD sense---from the true localization error distribution $P_{\varepsilon}$, then the MSE for the real data $P$ is also ``not too far"---in the sense of \eqref{strongupperbound}---from the MSE for the presumed data $Q$.

While \eqref{strongupperbound} is the strong version of the type of bound \cite{weiss2022bilateralbound}, it is instructive to analyze a looser version thereof (stated below), which nevertheless provides better insights as to the type of more ``tolerable" forms of mismatch. To this end,  we will make use of the following technical lemmas.
\vspace{-0.5cm}
\begin{lem}\label{lemma1}
Let $\mybar{\uh}_{Q}^{(k)}$, shorthand for $\mybar{\uh}_{Q}^{(k)}(\up,\mathcal{Q}_{\emph{\env}})$, be defined as\footnote{Recall that $\bar{h}_{\ell}[k]$, defined in \eqref{modelequationfreq}, depends on the environment $\mathcal{Q}_{\env}$ or $\mathcal{P}_{\env}$.}
\begin{equation}\label{newdeffreqvectors}
\mybar{\uh}_{Q}^{(k)}(\up,\mathcal{Q}_{\emph{\env}})\triangleq\left[\bar{h}_{1}[k]\;\cdots\;\bar{h}_{L}[k]\right]^{\tps}\in\Cset^{L\times 1},
\end{equation}
where $\{\bar{h}_{\ell}[k]=\ud_{\ell}^{\her}[k]\ub_{\ell}\}$ in \eqref{newdeffreqvectors} are for $\mathcal{Q}_{\env}$. Then the eigenvalues of
\begin{equation}\label{rankoneupdatedmatrix}
    \widetilde{\H}_{Q}^{(k)}\triangleq\sigma_s^2\mybar{\uh}_{Q}^{(k)}{\mybar{\uh}_{Q}^{(k)}}^{\her} + \sigma_v^2\I_L
\end{equation}
are given by
\begin{equation}\label{simpleformofeigenvalues}
    \ulambda^{(Q)}_k=[\sigma_s^2\|\mybar{\uh}_{Q}^{(k)}\|^2+\sigma_v^2 \;\; \underbrace{\sigma_v^2 \;\; \cdots \;\; \sigma_v^2}_{L-1\text{ \emph{terms}}}]^{\tps}\in\Rset^{L\times 1}_+.
\end{equation}
\end{lem}
\begin{proof}
The proof follows by identifying \eqref{rankoneupdatedmatrix} as a special rank-one ``updated matrix"\footnote{Using the same terminology as the one used \cite{ding2007eigenvalues}.} of $\sigma_v^2\I_L$, observing that $\sigma_s\mybar{\uh}_{Q}^{(k)}$ is an eigenvector of $\sigma_v^2\I_L$ with an eigenvalue $\sigma_v^2$, and applying \cite[Th. 2.1.]{ding2007eigenvalues} to get \eqref{simpleformofeigenvalues}.
\end{proof}
\begin{lem}\label{prop1}{\emph{\textbf{\cite[Ch.~7]{comon2010handbook}}}}
Assume $\det\left(\mSigma_{P}\right)\neq0$. Let $\Omeg\triangleq\mSigma_{Q}\mSigma_{P}^{-1}$, and denote $\Omeg\U=\U\Diag(\upsi)$ as the eigendecomposition of $\Omeg$. Then, if
\begin{equation*}
    \psi_m\neq\psi_n^*, \quad \forall m\neq n\in\{1,\ldots,NL\},
\end{equation*}
$\U$ is a joint diagonalizer (JD) of $(\mSigma_{Q},\mSigma_{P})$. In particular,
\begin{equation}\label{jointdiagonalization}
    \mSigma_{Q}=\U\Diag\left(\ulambda^{(Q)}\right)\U^{\her}, \mSigma_{P}=\U\Diag\left(\ulambda^{(P)}\right)\U^{\her},
\end{equation}
for $\ulambda^{(Q)}\triangleq[\ulambda_1^{(Q)} \cdots \ulambda_{N}^{(Q)}]^{\tps},\ulambda^{(P)}\triangleq[\ulambda_1^{(P)} \cdots \ulambda_{N}^{(P)}]^{\tps}\in\Rset_+^{NL\times1}$, where $\{\ulambda_{k}^{(Q)}\in\Rset_+^{L\times1}\}_{k=1}^N$ are defined in \eqref{simpleformofeigenvalues}, and $\{\ulambda_{k}^{(P)}\in\Rset_+^{L\times1}\}_{k=1}^N$ are defined similarly but for the environment $\mathcal{P}_{\emph{\env}}$.
\end{lem}
\noindent We now provide a generally looser, yet more insightful version of Prop.~\ref{prop0}.
\begin{theorem}\label{theorem1}
Assume the conditions of Lemma \ref{prop1} hold, and let $\U$ be the JD of $(\mSigma_{Q},\mSigma_{P})$ with $\ulambda^{(Q)},\ulambda^{(P)}$ as defined in \eqref{jointdiagonalization}. Then, if
\begin{equation}\label{simplifiedgammacondition}
\begin{gathered}
    2\left\|\mybar{\uh}_{Q}^{(k)}\right\|^2+\snr^{-1}>\left\|\mybar{\uh}_{P}^{(k)}\right\|^2, \quad \forall k\in\{1,\ldots,N\},
\end{gathered}
\end{equation}
where $\snr\triangleq\frac{\sigma_s^2}{\sigma_v^2}$, the MSE of $\hup$ for the true model $P$ is upper bounded by
\begin{align}
    \Eset_{P}\left[\|\uvarep_{\text{\boldmath $p$}}\|^2\right]&\leq \Eset_{Q}\left[\|\uvarep_{\text{\boldmath $p$}}\|^2\right]+\delta(P,Q),\label{weakerupperbound}
\end{align}
where $\delta(P,Q)\triangleq\sqrt{{\rm Var}_{Q}\left(\|\uvarep_{\text{\boldmath $p$}}\|^2\right)\Delta^2}$ with
\begin{equation}
\begin{aligned}\label{chisquaredivergencesimplified}
    &\Delta^2+1\triangleq\\
    &\prod_{k=1}^{N}\frac{\left(\snr\|\mybar{\uh}_{Q}^{(k)}\|^2+1\right)^2}{\left(\snr\|\mybar{\uh}_{P}^{(k)}\|^2+1\right)\left(\snr(2\|\mybar{\uh}_{Q}^{(k)}\|^2-\|\mybar{\uh}_{P}^{(k)}\|^2)+1\right)}.
\end{aligned}
\end{equation}
\end{theorem}
\begin{proof}
By applying \cite[Cor.~1]{weiss2022bilateralbound}, we obtain
\begin{align}
    \Eset_{P}\left[\|\uvarep_{\text{\boldmath $p$}}\|^2\right]&\leq \Eset_{Q}\left[\|\uvarep_{\text{\boldmath $p$}}\|^2\right]+\sqrt{{\rm Var}_{Q}\left(\|\uvarep_{\text{\boldmath $p$}}\|^2\right)\chi^2\left(P||Q\right)}.\label{upperbound}
\end{align}

Since in this case $P=\mathcal{CN}(\uo,\mSigma_{P})$, $Q=\mathcal{CN}(\uo,\mSigma_{Q})$, computation of $\chi^2\left(P||Q\right)$ by evaluating the integral implied by \cite[Eq.~4]{weiss2022bilateralbound} yields
\begin{equation}\label{chisquareofzeromeanGaussians}
    \chi^2\left(P||Q\right)=\frac{\det\left(\mSigma_{Q}\right)}{\det\left(\mSigma_{P}\right)\det\left(2\I_{NL}-\mSigma_{P}\mSigma^{-1}_{Q}\right)} - 1.
\end{equation}
By assumption, the conditions of Lemma \ref{prop1} hold. Therefore, substituting \eqref{jointdiagonalization} into \eqref{chisquareofzeromeanGaussians} readily gives
\begin{equation}\label{chisquaredivergenceclosedform}
    \Delta^2\triangleq\prod_{k=1}^{N}\prod_{\ell=1}^{L}\left(\frac{\gamma^2[k,\ell]}{2\gamma[k,\ell]-1}\right)-1,
\end{equation}
under the condition
\begin{equation}\label{conditiontheorem1}
\begin{gathered}
\hspace{-0.1cm}\gamma[k,\ell]\triangleq\frac{\lambda^{(Q)}_{k}[\ell]}{\lambda^{(P)}_{k}[\ell]}>\frac{1}{2}, \, \forall k\in\{1,\ldots,N\}, \, \forall \ell\in\{1,\ldots,L\}.
\end{gathered}
\end{equation}
This arises because, after joint diagonalization, the entries of the transformed complex-normal (CN) vectors become uncorrelated, and due to their circular symmetry, they are statistically independent. Therefore, by virtue of the tensorization property of the CSD, each term of the product \eqref{chisquaredivergenceclosedform} corresponds to a CSD (plus one) of two \emph{univariate} zero-mean CN random variables, which must be non-negative.

The proof now continues with two phases. In \emph{Phase 1}, we obtain closed-form expressions for the eigenvalues $\ulambda^{(Q)},\ulambda^{(P)}$. In \emph{Phase 2}, we express $\{\gamma[k,\ell]\}$ in terms of these eigenvalues, to obtain \eqref{simplifiedgammacondition} and \eqref{chisquaredivergencesimplified}.

\underline{\emph{Phase 1}}: The elements $\{\lambda^{(Q)}_{k}[\ell],\lambda^{(P)}_{k}[\ell]\}$ of $\{\ulambda_k^{(Q)},\ulambda_k^{(P)}\}$ can be efficiently computed, as follows. Recall that $\{\H_{\ell}\}$ are all diagonal. Hence, the covariance matrices $\mSigma_{Q}$ and $\mSigma_{P}$ \eqref{defofcovariancematrices} have a block structure, where each $N\times N$ block is a diagonal matrix. Thus, by applying the appropriate permutation matrix $\PI$ and its inverse $\PI^{\tps}$ on both sides, we have (focusing on $\mSigma_{Q}$, since the logical argument is identical for $\mSigma_{P}$)
\begin{equation*}
    \mSigma_{Q}=\PI_{\row}^{\tps}\Big(\PI_{\row}\mSigma_{Q}\PI_{\col}\Big)\PI^{\tps}_{\col}\triangleq\PI_{\row}^{\tps}\widetilde{\mSigma}_Q\PI^{\tps}_{\col},
\end{equation*}
where $\PI_{\row}, \PI_{\col}$ permute the rows and columns, respectively, and where
\begin{align*}
    \widetilde{\mSigma}_Q &=\Blkdiag\left(\widetilde{\H}_{Q}^{(1)},\ldots,\widetilde{\H}_{Q}^{(N)}\right)\in\Cset^{NL\times NL},
\end{align*}
and $\Blkdiag(\cdot,\ldots,\cdot)$ forms a block diagonal matrix from its (square) matrix arguments. Since any permutation matrix is a unitary matrix, $\det\left( \mSigma_{Q}\right)=\det(\widetilde{\mSigma}_{Q})$. Moreover, since $\widetilde{\mSigma}_{Q}$ is a block diagonal matrix, if follows that, for every $k$, $\ulambda_{k}^{(Q)}$ are the eigenvalues of $\widetilde{\H}_{Q}^{(k)}$ \eqref{rankoneupdatedmatrix}. Using the same arguments, $\ulambda_{k}^{(P)}$ are the eigenvalues of $\widetilde{\H}_{P}^{(k)}$. We conclude that $\ulambda^{(Q)}$, $\ulambda^{(P)}$ can be computed via $N$ eigendecomposition of $L\times L$ matrices, rather than a single (na\"ive) eigendecomposition of an $NL\times NL$ matrix. To complete phase 1, we have from Lemma \ref{lemma1} that
\begin{equation}\label{closedformexpressionsofeigenvalues}
    \lambda_k^{(D)}[\ell]=\begin{cases}
    \sigma_s^2\|\mybar{\uh}_{D}^{(k)}\|^2+\sigma_v^2, & \ell=1,\\
    \sigma_v^2, & 2\leq\ell\leq L,
    \end{cases} \; D\in\{Q,P\}.
\end{equation}

\underline{\emph{Phase 2}}:
Using the closed-form expressions \eqref{closedformexpressionsofeigenvalues}, we now have
\begin{equation}\label{gammaclosedform}
    \gamma[k,\ell]=\frac{\lambda^{(Q)}_{k}[\ell]}{\lambda^{(P)}_{k}[\ell]}=\begin{cases}
    \frac{\sigma_s^2\left\|\mybar{\text{\scriptsize \boldmath$h$}}_{Q}^{(k)}\right\|^2+\sigma_v^2}{\sigma_s^2\left\|\mybar{\text{\scriptsize \boldmath$h$}}_{P}^{(k)}\right\|^2+\sigma_v^2}, & \ell=1,\\
    1, & 2\leq\ell\leq L.
    \end{cases}
\end{equation}
Therefore, substituting \eqref{gammaclosedform} into \eqref{chisquaredivergenceclosedform}, the latter is further simplified into
\begin{align*}
    &\Delta^2+1=\prod_{k=1}^{N}\prod_{\ell=1}^{L}\left(\frac{\gamma^2[k,\ell]}{2\gamma[k,\ell]-1}\right)=\prod_{k=1}^{N}\left(\frac{\gamma^2[k,1]}{2\gamma[k,1]-1}\right)\\
    &=\prod_{k=1}^{N}\frac{\left(\sigma_s^2\|\mybar{\uh}_{Q}^{(k)}\|^2+\sigma_v^2\right)^2}{\left(\sigma_s^2\|\mybar{\uh}_{P}^{(k)}\|^2+\sigma_v^2\right)\left(\sigma_s^2(2\|\mybar{\uh}_{Q}^{(k)}\|^2-\|\mybar{\uh}_{P}^{(k)}\|^2)+\sigma_v^2\right)}.
\end{align*}
Using $\snr=\frac{\sigma_s^2}{\sigma_v^2}$, we obtain \eqref{chisquaredivergencesimplified}, and with \eqref{gammaclosedform}, \eqref{conditiontheorem1} simplifies to \eqref{simplifiedgammacondition}.
\end{proof}

It is instructive to evaluate \eqref{chisquaredivergencesimplified} in the following signal-to-noise ratio (SNR) limiting cases. 

\textbf{High SNR Regime:} To better understand what governs the gap $\delta(P,Q)$ at high SNR, define $\rho_k\triangleq\|\mybar{\uh}_{P}^{(k)}\|^2/\|\mybar{\uh}_{Q}^{(k)}\|^2$, and observe that
\begin{equation*}
    \lim_{\snr\to\infty}\Delta^2=\prod_{k=1}^{N}\frac{1}{\rho_k(2-\rho_k)}-1.
\end{equation*}
Thus, the gap is governed by product of functions of all the \emph{per-frequency} energy deviations of the true frequency response from the presumed one.

\textbf{Low SNR Regime:} In this case, for any set of finite $\{\rho^{(k)}\}$ per-frequency energy deviations, at low SNR, it follows from \eqref{gammaclosedform} that
\begin{equation*}
    \forall k\in\{1,\ldots,N\}:\,\lim_{\snr\to0^{+}}\gamma[k,1]=1 \Longrightarrow \lim_{\snr\to0^{+}}\Delta^2=0.
\end{equation*}
Hence, in the low SNR limit, the model mismatch is essentially irrelevant (in the sense of the attainable MSE). 

A ``conceptual interpolation" between these two extreme cases suggests that as the SNR increases, the model mismatch has an increasingly stronger effect on the gap $\Delta^2$, and consequently on the gap $\delta(P,Q)$.



\vspace{-0.3cm}
\subsection{Simulation Results: DLOC in a Mismatched Environment}\label{sec:results}
\vspace{-0.1cm}
\begin{figure}[t]
	\includegraphics[width=0.475\textwidth]{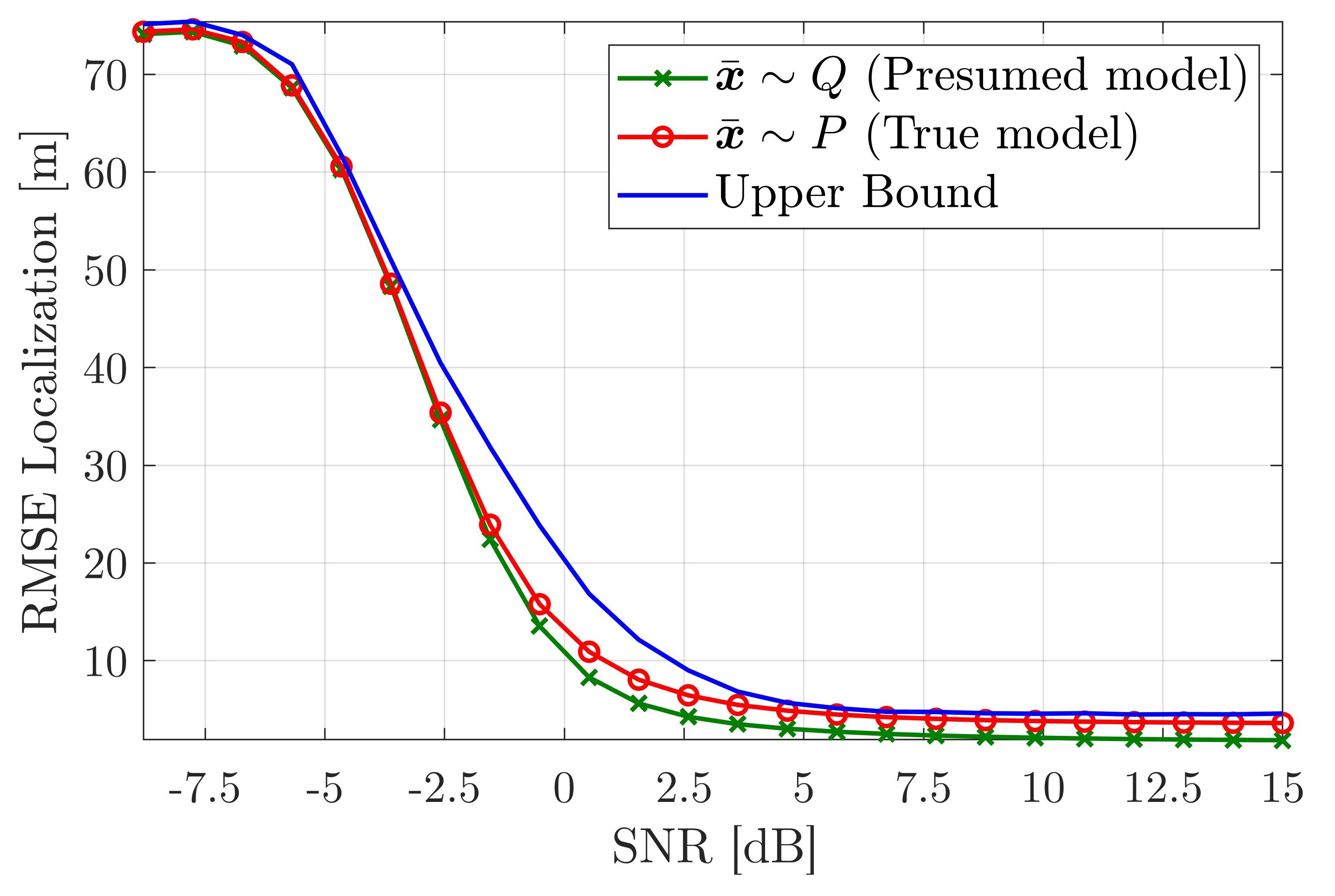}
	\centering\vspace{-0.3cm}
	\caption{Localization RMSE vs.\ SNR for the presumed and true models, $Q$ and $P$, respectively, and the upper bound \eqref{strongupperbound}.}
	\label{fig:results}\vspace{-0.5cm}
\end{figure}
We consider the same scenario described in \cite{weiss2022direct} with $L=4$ receivers, and a single source whose position was drawn (once, and then fixed) from a 3-dimensional uniform distribution supported on the volume of interest. However, instead of the $3$-ray propagation model as in \cite{weiss2022direct}, we now use the Bellhop simulator to simulate a substantially richer environment, and generate the resulting CIR of each source-receiver pair corresponding to a pre-specified set of environmental parameters. For the training dataset $\mathcal{D}^{(J)}_{\text{train}}$, we use the environment $\mathcal{Q}_{\env}$ (similar to the one illustrated in Fig.~\ref{fig:example}), and for the test dataset we use $\mathcal{P}_{\env}$, in which we use a \emph{different} depth-varying speed of sound profile. The technical details pertaining to these environments are deferred to the supplementary material \cite{weiss2022icasspsuppmat}.

Fig.~\ref{fig:example} shows the root MSE (RMSE) vs.\ the SNR\footnote{More accurately, a normalized SNR, relative to the average CIR attenuation.} obtained by our CNN-based solution for: (i) a testset that is $Q$-distributed, the same as the training set, from the environment $\mathcal{Q}_{\env}$, and; (ii) a testset that is $P$-distributed, stemming from the mismatched environment $\mathcal{P}_{\env}$. We also plot the upper bound \eqref{strongupperbound}, which was evaluated numerically.\footnote{The CSD was estimated using the method in \cite{poczos2011estimation,poczos2012nonparametric}, see also \cite{ryu2022nearest}.} The results are based on averaging $10^6$ independent trials per SNR point. 

First, the robustness of the proposed DD-DLOC solution is evident, as reflected from the graceful performance degradation of the red curve relative to the green. Further, the upper bound \eqref{strongupperbound} is not only fairly tight at the low and high SNR regimes, but is also informative about the SNR region in which the threshold phenomenon takes place. These results demonstrate the successful operation of our DD-DLOC solution in a rich and highly reverberant UWA environment, and at the same time corroborate our theoretical result in Prop.~\ref{prop0}.

\vspace{-0.3cm}
\section{Concluding Remarks}\label{sec:conclusion}
\vspace{-0.2cm}
The theoretical results of this work have an operational interpretation and an important implication to the DD-DLOC problem. They suggest that a system with a DNN that was trained on a dataset representing one environment ($\mathcal{Q}_{\env}$), which gives rise to a collection of possible CIRs, may still perform well (in the sense of the upper bounds \eqref{strongupperbound} or \eqref{upperbound}) if it is deployed in a different environment ($\mathcal{P}_{\env}$) that gives rise to a \emph{different} collection of possible CIRs, as long as the overall mismatch is not too large (e.g., \eqref{simplifiedgammacondition}). This has been demonstrated when using the deep CNN discussed in Section \ref{sec:learning2localize} for a simulated reverberant UWA environment.

An important aspect to be studied through future lines of this work is the domain-specific, informed incorporation of randomness to the training dataset. Such randomness should be selected so as to balance between localization accuracy and the MSE-robustness to CIR deviations.

\bibliographystyle{IEEEbib}
\small{\bibliography{refs}}

\end{document}